\documentclass{amsart}%
\usepackage{cite}
\usepackage{hyperref}
\usepackage{amssymb}
\usepackage{amsmath}
\usepackage{amsthm}
\usepackage{amsfonts}
\usepackage{graphicx}
\usepackage{bbm}
\usepackage{color}
\usepackage{enumerate}
\usepackage{fullpage}
\usepackage{subcaption}
\usepackage{float}
\usepackage{enumitem}
\usepackage[normalem]{ulem}

\newtheorem{thm}{Theorem}[section]

\newtheorem{asmp}{Assumption}
\newtheorem{lemma}[thm]{Lemma}

\newtheorem{prop}[thm]{Proposition}

\newtheorem{theorem}[thm]{Theorem}
\theoremstyle{definition}

\newtheorem{example}{Example}[section]
\newtheorem{rem}[thm]{Remark}

\numberwithin{equation}{section}
\newcommand{\Var}{\mathrm{Var}}

\newtheorem*{acknowledgement}{Acknowledgement}

\begin{document}

	\title[Black-Scholes]{A derivation of the Black-Scholes option pricing model using a central limit theorem argument}

	\author[Majumdar]{Rajeshwari Majumdar{$^{\S\dag}$}}
	\thanks{\footnotemark {$\S$} Research was supported in part by NSF Grant DMS-1262929.}
	\address{$^{\dag}$ Department of Politics\\
		New York University\\
                    New York, NY 10012, U.S.A}
	\email{majumdar@nyu.edu}
	
	\author[Mariano]{Phanuel Mariano{$^{\S\ddag}$}}
%	\thanks{\footnotemark {$\dag$} Research was supported in part by NSF Grant DMS-1262929.}
	\address{$^{\ddag}$ Department of Mathematics\\
		Purdue University\\
		West Lafayette, IN 47907,  U.S.A.}
	\email{pmariano@purdue.edu}
	
	\author[Peng]{Lowen Peng{$^{\S\star}$}}
%	\thanks{\footnotemark {$\dag$} Research was supported in part by NSF Grant DMS-1262929.}
	\address{$^{\star}$ Department of Mathematics\\
		University of Connecticut\\
		Storrs, CT 06269,  U.S.A.}
	\email{lowen.peng@uconn.edu}
	
	\author[Sisti]{Anthony Sisti{$^{\S\star\star}$}}
%	\thanks{\footnotemark {$\dag$} Research was supported in part by NSF Grant DMS-1262929.}
	\address{$^{\star\star}$ Department of Mathematics\\
		University of Connecticut\\
		Storrs, CT 06269,  U.S.A.}
	\email{anthony.sisti@uconn.edu}

	\keywords{Black-Scholes; mathematical finance; options pricing; Central Limit Theorem}
	
	\subjclass{Primary 91B28; Secondary 91G20 60G99 60F05 }

%	\date{\today \ \emph{File:\jobname{.tex}}}

	\begin{abstract} 
The Black-Scholes model (sometimes known as the Black-Scholes-Merton model) gives a theoretical estimate for the price of European options. The price evolution under this model is described by the Black-Scholes formula, one of the most well-known formulas in mathematical finance. For their discovery, Merton and Scholes  have been awarded the 1997 Nobel prize in Economics. The standard method of deriving the Black-Scholes European call option pricing formula involves stochastic differential equations. This approach is out of reach for most students learning the model for the first time. We provide an alternate derivation using the Lindeberg-Feller central limit theorem under suitable assumptions. Our approach is elementary and can be understood by undergraduates taking a standard undergraduate course in probability.

	\end{abstract}

	\maketitle

	\tableofcontents

	\renewcommand{\contentsname}{Table of Contents}

	\section{Introduction}\label{s.1}

\indent The Black-Scholes model was proposed by Fischer Black and Myron Scholes in their 1973 paper entitled ``The Pricing of Options and Corporate Liabilities''. They derived a formula for the value of a ``European-style" option in terms of the price of the stock \cite{bs1973} by utilizing techniques from stochastic calculus and partial differential equations. Later in 1973, Robert C. Merton expanded the mathematical ideas underlying the Black-Scholes model in his paper entitled ``Theory of Rational Option Pricing'' \cite{Merton1973}. Since its introduction, the formula has been widely used by option traders to approximate prices and has lead to a variety of new models for pricing derivatives. 

A modern derivation of the Black-Scholes model can be found in\cite{Ross2011}. The derivation given by Ross uses over 100 pages to arrive to the Black-Scholes and requires a discussion on geometric Brownian motion. In this paper, we consider an alternative approach to the derivation of the Black-Scholes European call option pricing formula using the central limit theorem. Our approach will be concise, elementary and can be understood by anyone taking a standard undergraduate course in probability.

The central limit theorem has played a crucial role in the development of modern probability theory, with Laplace, Poisson, Cauchy, Lindeberg, and L\'evy among the mathematicians who have contributed to its development in the nineteenth and twentieth centuries. The basic form of the central limit theorem, as given in \cite{PetrovBook1995}, is as follows: the sum of a sufficiently large number of independent and identically distributed random variables with finite mean and variance approximates a normal random variable in distribution. Mathematically, let $X_1, X_2, \dots,$ be a sequence of random variables and $S_n = \sum_{k=1}^n X_k$. Then, under a variety of different conditions, the distribution function of the appropriately centered and normalized sum $S_n$ converges to the standard normal distribution function as $n \rightarrow \infty$. In Section 3, we use the Lindeberg-Feller variant of the central limit theorem (as stated in \cite{Varadhan}) to establish the log-normality of the asset price under suitable assumptions:

\begin{theorem}[\textbf{Lindeberg-Feller}]\label{lindFeller}
Suppose for each $n$ and  $i=1,\dots n$, $X_{ni}$
are independent and have mean $0$. Let $S_{n}=\sum_{i=1}^{n}X_{ni}$. 
Suppose that $\sum_{i=1}^{n}\mathbb{E}[X_{ni}^{2}]\to\sigma^{2}$ for $0<\sigma^{2}<\infty$. 				
Then, the following two conditions are equivalent: 
		%[label=(\alph*)]
\begin{enumerate}
\item[(a)]  $S_{n}$ converges weakly to a normal random variable with mean $0$ and variance $\sigma^{2}$, and the triangular array $\left\{ X_{ni}\right\}$ 
satisfies the condition that 
\begin{equation*}\label{uanStatement}
\lim_{n\rightarrow0}\max_{i}\mathbb{E}\left(X_{ni}^{2}\right)=0.
\end{equation*}
\item[(b)] (Lindeberg Condition) For all $\epsilon > 0$, $$\sum_{i=1}^{n}\mathbb{E}\left[X_{ni}^{2};\left|X_{ni}\right|>\epsilon\right]\to0.$$
\end{enumerate}

\end{theorem} 

Our work in this paper is inspired by Chapters 17 and 18 of \cite{Sengupta}. We give a rigorous mathematical treatment of the results discussed in that text using an elementary approach that is accessible to students who have taken an undergraduate probability course. The rest of this section introduces the basic financial concepts underlying the Black-Scholes model.

	\indent A \textbf{financial instrument} is any asset that can be traded on the market. Consider the following kind of instrument: If an event $B$ occurs, the holder of the instrument receives one dollar, and if $B$ does not occur, the holder receives nothing. The value of such an instrument is dependent on the probability that the event occurs. This probability is assessed through a \textbf{pricing measure}, denoted by $Q$. A pricing measure can be understood as a way to determine the amount of the underlying asset that one would be willing to pay in order to own a financial instrument. For example, if a financial instrument involves the exchange of one dollar given the event $B$ occurs, and the probability that event $B$ occurs is $Q(B)$, an individual would be willing to risk $Q(B)$ dollars to own the instrument. \\
	\indent A measuring unit for the price of a financial instrument is called a \textbf{numeraire}. In the previous example, the dollar would function as a numeraire and the pricing measure would be with respect to dollars. Numeraires have time stamps, so their value corresponds to a set date. Consider numeraires such as one unit of cash today, or one unit of cash at a future time $t$; the value of that unit of cash may differ from today to that time $t$. Thus, we specify that a pricing measure is with respect to the numeraire unit cash at time-$t$.\\
	\indent A call (respectively, put) \textbf{option} is a contract that gives the option holder the right to buy (respectively, sell) an asset for a certain price $K$, called the \textbf{strike price}, during the time period $[0,t]$ (for an American option) or at time $t$ (for a European option), where $t$ is the expiration time of that right (often referred to as just the \textbf{expiration time}). In what follows, we price a \textbf{European call option}, which entitles the holder to purchase a unit of the underlying asset at expiration $t$ for strike $K$. 
	
	The Black-Scholes formula for the price of a European call option is derived
under the assumption that there is no arbitrage opportunity surrounding a
trade of the option (or the underlying instrument), that is, one cannot expect
to generate a risk-free profit by purchasing (or selling) the option (or the
underlying instrument). We denote the time when the option is priced as
time $0$, when the underlying instrument is valued at $X_0$. Recall the option
expires at time $t$ and the strike price for the option is $K$. Suppose the
risk-free rate of interest is $r$. If the option is priced for $C$, then the
future-value of it at time $t$, under continuous compounding, is $Ce^{rt}$.
With $X_t$ denoting the price of the underlying instrument at time $t$, the
payoff of the option is $\max\left(X_t-K,0\right)$. The no arbitrage
opportunity on the option trade requires the equation
\begin{equation}\label{noarboption}
C=e^{-rt}\mathbb{E}\left(\max\left(X_t-K,0\right)\right),
\end{equation}
to hold; similarly, the no arbitrage opportunity on the trade of the
underlying instrument requires the equation
\begin{equation}\label{noarbsec}
\mathbb{E}\left(X_t\right)=X_0e^{rt},
\end{equation}
to hold.
The formula used to price the European call option under the \textbf{Black-Scholes European option pricing model} is given by 
%	\begin{equation}\label{eqnPrice}
%	C = X_0N(d_+) - Ke^{-rt}N(d_-)
%	\end{equation}
	\begin{equation}\label{eqnPrice}
	C = X_0N(d_+) - Ke^{-rt}N(d_-),
	\end{equation}
	where $N$ is the standard Normal CDF, that is, 
\[
N(x)=\frac{1}{\sqrt{2\pi}}\int_{-\infty}^{x}e^{-y^{2}/2}dy,
\]
\[d_\pm=\frac1{\sigma \sqrt{t}}\log\left[e^{rt}X_0/K\right]\pm\frac12\sigma \sqrt{t},\]
and 
$\sigma$ is the volatility of the return on the underlying asset through expiration.

	\begin{example} Consider the pricing of a European call option on a stock with a present value of 50 Euros and a strike price of 52 Euros under the following conditions: $r$ = 4\% (per annum), $t$ = 1 (year), $\sigma$ = 0.15. To calculate the price of this option we use Equation \eqref{eqnPrice}. We first find 

		$$d_+ = \frac{\log\left[e^{0.04(1)}50/52\right]}{0.15} + \frac{1}{2}(0.15)= 0.0802  $$
		and
		
		$$d_- = \frac{\log\left[e^{0.04(1)}50/52\right]}{0.15} - \frac{1}{2}(0.15)= -0.0698; $$
		\ \\
		we then have
		\begin{align*}
		C
		& = 50N(0.0802) - 52e^{-(0.04)(1)}N(-0.0698) \\
		& =50(.532) - 52(0.96)(0.472) \\
		& = 3.04.
		\end{align*}
		Thus, from the Black-Scholes model, the price of this call option would be 3.04 Euros.
	\end{example}

In section 2, we derive the call option pricing formula assuming the log-normality of the underlying asset price. In section 3, we prove the log-normality under suitable assumptions.

	\section{Pricing the European Call Option}\label{s.2}
To derive the call option pricing formula in \eqref{eqnPrice} we first show the following fact regarding normal random variables.

\begin{lemma}\label{Sengupta1711}
For any normal random variable $Y$ with mean $\mu_Y$, standard deviation $\sigma_Y$, and $M>0$, we have
\begin{equation*}
\mathbb{E}\left(\max\left(e^Y-M,0\right)\right)=\mathbb{E}\left(e^Y%
\right)N\left(h_{+}\right)-MN\left(h_{-}\right)\text{,}
\end{equation*}
where 
\begin{equation*}
h_\pm=\left[\log\left(\mathbb{E}\left(e^Y\right)\Big/M\right)\pm\frac{1}{2}%
\sigma_Y^2\right]\Big/\sigma_Y.
\end{equation*}

\end{lemma}	

\begin{proof}
For a normal
random variable $Y$,
\begin{equation*}
\mathbb{E}\left(e^Y\right)=e^{\mu_Y+\frac{\sigma_Y^2}{2}}\text{.}
\end{equation*}
As such,
\begin{equation*}
\begin{aligned}[t]
&h_{+}=\frac{\mu_Y+\sigma_Y^2-\log M}{\sigma_Y}\\
&h_{-}=\frac{\mu_Y-\log M}{\sigma_Y}\vcenter{\hbox{.}}
\end{aligned}
\end{equation*}
Now note that
\begin{equation*}
\mathbb{E}\left(\max\left(e^Y-M,0\right)\right)=\int_{\log M}^\infty
e^y\phi_{\mu_Y,\sigma_Y}\left(y\right)dy-MP\left(Y>\log M\right)\text{,}
\end{equation*}
where $\phi_{\mu_Y,\sigma_Y}$ is the density of $Y$. Completing the square
\begin{equation*}
y-\frac{\left(y-\mu_Y\right)^2}{2\sigma_Y^2}=\mu_Y+\frac{\sigma_Y^2}{2}-%
\frac{\left(y-\left(\mu_Y+\sigma_Y^2\right)\right)^2}{2\sigma_Y^2}
\end{equation*}
we obtain, using the identity $1-N\left(x\right)=N\left(-x\right)$, 
\begin{equation*}
\int_{\log M}^\infty
e^y\phi_{\mu_Y,\sigma_Y}\left(y\right)dy=e^{\mu_Y+\frac{\sigma_Y^2}{2}}N%
\left(h_{+}\right).
\end{equation*}
Since
\begin{equation*}
P\left(Y>\log M\right)=N\left(h_{-}\right)\text{,}
\end{equation*}
the equality follows.
\end{proof}

	Recall that under the assumption of no arbitrage, the price of a European call option must equal the expected payoff of the option. Expectation is computed with respect to the pricing measure $Q_t$, corresponding to time-$t$ cash numeraire.

	\begin{prop}
	Assume there are no opportunities for arbitrage and the risk-free interest rate is $r$. Consider a European call option on an instrument with expiration $t$ and strike $K$. Let $X_t$ be the time-$t$ price of the underlying instrument, where $X_t=X_0e^{Y_t}$ and the $Q_t$-induced distribution of $Y_t$ is $ \mathcal N(\mu_{Y_t},\sigma_{Y_t}^2)$. Then, the discounted (that is, time-0) price of the call option,  $C$, is given by
	 \begin{equation}\label{eqnPrice2}
	C = X_0N(d_+) - Ke^{-rt}N(d_-),
	\end{equation}
	where
	\[d_\pm=\frac1{\sigma_{Y_t}}\log\left[e^{rt}X_0/K\right]\pm\frac12\sigma_{Y_t}.\]

		%	\begin{equation}\label{changeThis}
		%	C=e^{-rt}\mathbb E_{Q_t}\left[\max(X_t-K,0)\right]=X_0\Phi(d_+)-e^{-rt}K\Phi(d_-)
		%	\end{equation}
		%	where $\Phi$ is the standard normal CDF and
		%	\[d_\pm=\frac1\sigma\log\left[e^{rt}X_0/K\right]\pm\frac12\sigma\]
	\end{prop}
	
\begin{proof}
From the definition of $X_t$,
\begin{equation*}
\max\left(X_t-K,0\right)=X_0\max\left(e^{Y_t}-\frac{K}{X_0}\vcenter{\hbox{$,%
$}} \ 0\right).
\end{equation*}
By Lemma \ref{Sengupta1711},
\begin{equation*}
\mathbb{E}\left(\max\left(e^{Y_t}-\frac{K}{X_0}\vcenter{\hbox{$,$}}\ 0\right)%
\right)=\mathbb{E}\left(e^{Y_t}\right)N\left(h_{+}\right)-\frac{K}{X_0}N%
\left(h_{-}\right)
\end{equation*}
with
\begin{equation*}
h_\pm=\left[\log\left(\mathbb{E}\left(e^{Y_t}\right)\frac{X_0}{K}\right)\pm%
\frac{1}{2}\sigma_{Y_t}^2\right]\Big/\sigma_{Y_t}=d_\pm\text{,}
\end{equation*}
where the second equality follows from Equation \eqref{noarbsec}. The
proof follows from Equation \eqref{noarboption}.
\end{proof}

	\section{Log-Normality of Prices}\label{s.3}
	
	In the previous section, we derived the Black-Scholes formula on the premise that our prices follow a log-normal distribution. In this section, we use the Lindeberg-Feller central limit theorem to prove this premise under the following assumptions.  Expectation is computed with respect to the pricing measure $Q_0$, corresponding to time-$0$ cash numeraire.
	
	\begin{asmp}\label{A1}
		For each $t$, the random variable $Y_t= \log \frac{X_t}{X_0}$ has finite variance.
	\end{asmp}
	
	\begin{asmp}\label{A2}
		The process $Y_t$ has stationary and independent increments. That is, the differences $Y_t-Y_s$ are independent for disjoint intervals $[s,t]$; for intervals of equal length, they are i.i.d.
	\end{asmp}

	\begin{asmp}\label{A4}
		For every $\epsilon > 0$, 
		$n\mathbb{E}\left[\left(Y_{t/n}-Y_{0}\right)^{2} ; \left|Y_{t/n}-Y_{0}\right|>\epsilon\right]\rightarrow0$ as $n\to\infty$. 
	\end{asmp}

	\begin{theorem}\label{lognormtheorem}
		Under Assumptions 1, 2, and 3, for every $t>0$, $Y_t$ is a normal random variable with respect to the pricing measure $Q_0$ with variance $\sigma^2 t$ for some constant $\sigma \geq 0$ and all $t\geq 0$.
	\end{theorem}

In addition to the Lindeberg-Feller central limit theorem (Theorem \ref{lindFeller}), the proof of Theorem \ref{lognormtheorem} makes use of the following lemma.

	\begin{lemma}\label{lemma3}
		Suppose $f:\left[ 0,\infty \right) \to \left[ 0,\infty \right)$ satisfies $f\left(x+y\right)=f(x)+f(y)$. There exists a constant $C$ such that $f(x)=Cx$ for all $x\geq 0$. 
	\end{lemma}
	
	\begin{proof}
		\ We first observe that
		\begin{equation*}
		f(0)=f(0+0)=f(0)+f(0)=2f(0), 
		\end{equation*}
		implying that $f(0)=0$. We can prove by induction that 
		\begin{equation}\label{lp1}
		f(m)=mf(1)\text{ for all }m\geq 1.
		\end{equation}
		Let $p$ and $q$ ($>1$) be positive integers with no common factors. By induction again,
		\begin{equation}\label{lp2}
		f\bigg(m\frac{p}{q}\bigg) =mf\bigg(\frac{p}{q}\bigg) \text{ for all }m\geq 1.
		\end{equation}
		Using Equation (\ref{lp1}) followed by Equation (\ref{lp2}),
		\begin{equation*}
		pf(1) = f(p) = f\bigg(q\frac{p}{q}\bigg) = qf\bigg(\frac{p}{q}\bigg);
		\end{equation*}
		hence, for every positive rational number $r$ of the form $p/q$,
		\begin{equation}\label{lp3}
		f(r) = f\bigg(\frac{p}{q}\bigg) = \frac{p}{q}f(1) = rf(1).
		\end{equation}
		Thus, we have shown that for every rational number $x$ in
		$\big[0,\infty\big)$,
		\begin{equation}\label{lp4}
		f(x)=Cx \text{ where }C=f(1).
		\end{equation} 

Note that if $x\leq y$, then
		\begin{equation*}
		f\left(y\right)=f\left(x+y-x\right)=f\left(x\right)+f\left(y-x\right)\geq
		f\left(x\right)\text{,}
		\end{equation*}
		showing that $f$ is non-decreasing.

We claim that Equation $\left(\text{\ref{lp4}}\right)$ holds for a positive irrational number
$d$. Let $n_0$ be a positive integer such that for all $n\geq n_0$,
\begin{equation*}
\frac{1}{n}<d.
\end{equation*}
For every $n\geq n_0$, choose
$r_n\in\left(d-\frac{1}{n}\vcenter{\hbox{$,d$}}\right)$ and
$s_n\in\left(d,d+\frac{1}{n}\right)$ to be arbitrary rational numbers. Then, by
Equation $\left(\text{\ref{lp4}}\right)$ and the observed
monotonicity of $f$,
\begin{equation*}
r_nf\left(1\right)=f\left(r_n\right)\leq f\left(d\right)\leq
f\left(s_n\right)=s_nf\left(1\right)\text{.}
\end{equation*}
Since $r_n \rightarrow d$ and $s_n \rightarrow d$, by the squeeze theorem we have that $f(d)$ converges to $df(1)$, as needed.

\end{proof}
	
We are now ready to prove Theorem \ref{lognormtheorem}.	
	
\begin{proof}
	We first show that
	\begin{equation}\label{varform}
	\Var\left[Y_t\right] = \sigma^2 t.
	\end{equation}
	To that end, note that 
	\begin{equation}
	Y_{t+s} - Y_0 = Y_{t+s}-Y_t + Y_t - Y_0;
	\end{equation}
	by Assumption \ref{A2}, independent increments followed by stationary increments, 
	\begin{align}\label{varexpand}
	\Var\left[Y_{t+s} - Y_0\right] &= \Var\left[ Y_{t+s}-Y_t\right] + \Var\left[Y_t - Y_0\right]  \nonumber\\
	&= \Var\left[Y_s-Y_0\right]+ \Var\left[Y_t - Y_0\right].
	\end{align}

	With $f(u)=\Var\left[Y_u - Y_0\right]$, Equation (\ref{varexpand}) reduces to 
	\begin{equation}
	f(t+s) = f(t) + f(s).
	\end{equation}
	Since $f$ is non-negative, by Lemma \ref{lemma3}, $f(t) = tf(1)$, where $f(1) = \Var\left[Y_1 - Y_0\right] = \sigma^2,$ 
	thus establishing Equation (\ref{varform}).

	Now, to prove the assertion of the theorem, we show that $Y_{t}-Y_{0}$, where $Y_0$ is a deterministic quantity, is normally distributed with variance $\sigma^2 t$ using the Lindeberg-Feller Theorem.  
	
	With 
	\begin{equation} \label{xni}
	X_{ni}=Y_{ti/n}-Y_{t(i-1)/n},
	\end{equation}
we obtain, by telescopic cancellation, 
	\begin{equation}
	Y_{t}-Y_{0}=\sum_{i=1}^{n}X_{ni}, \label{telescope}
	\end{equation}
	where the dependence of $X_{ni}$ on $t$ is suppressed for notational convenience. Since
	\[
	Y_{t}-Y_{0}-\mathbb{E}\left[Y_{t}-Y_{0}\right]=\sum_{i=1}^{n}\left[Y_{ti/n}-Y_{t(i-1)/n}-\left(\mathbb{E}\left[Y_{ti/n}\right]-\mathbb{E}\left[Y_{t(i-1)/n}\right]\right)\right],
	\]
	without loss of generality, we can assume that $Y_{t}-Y_{0}$ and
	$X_{ni}=Y_{ti/n}-Y_{t(i-1)/n}$  have mean zero.
By stationary increments in Assumption 2,
	$X_{ni}$ has the same distribution as $Y_{t/n}-Y_{0}$. Consequently, by Equation (\ref{varform}), 
	\begin{equation} \label{expxni2}
	\mathbb{E}\left[X_{ni}^2\right]=\sigma^{2}\frac{t}{n},
	\end{equation}
implying $\sum_{i=1}^{n}\mathbb{E}\left[X_{ni}^2\right] = \sigma^2 t.$ Thus we can apply the Lindeberg-Feller central limit theorem once the Lindeberg condition is satisfied.  
	
	Let $\epsilon>0$. By the consequence of the assumption of stationary increments noted above, 
	\begin{equation*}
	\sum_{i=1}^{n}\mathbb{E}\left[X_{ni}^{2};\left|X_{ni}\right|>\epsilon\right]  = n\mathbb{E}\left[\left(Y_{t/n}-Y_{0}\right)^{2};\left|Y_{t/n}-Y_{0}\right|>\epsilon\right],
	\end{equation*}
whence the Lindeberg condition follows from Assumption~\ref{A4}.

\end{proof} 

A couple remarks are in order.

\begin{rem}
While Assumptions \ref{A1} and \ref{A2} reflect reasonable properties of the asset price process, it is difficult to interpret Assumption \ref{A4}. It seems that the only significance of this assumption is its sufficiency for the Lindeberg condition. However, note that the array defined in Equation \eqref{xni}, by virtue of Equation \eqref{expxni2}, satisfies the second part of the first condition in the Lindeberg-Feller Theorem, rendering Assumption \ref{A4} necessary for the desired asymptotic normality.
\end{rem}

\begin{rem}
We note that Lindeberg's condition is needed, in principle, to avoid jumps in the stochastic process $Y_t$. Without this condition one
can obtain the Poisson process as a limit (or more generally a L\'evy
process), but this is outside of the scope of this article. See for example \cite[Theorem 28.5]{Billingsley1986}.

\end{rem}

\begin{acknowledgement}
The authors are grateful for many helpful and motivating conversations with M. Gordina, O. Mostovyi, H. Panzo, A. Sengupta and A. Teplyaev.

\end{acknowledgement}

	\bibliography{black_scholes_actual}
	\bibliographystyle{plain}
	\nocite{*}

\end{document}